\documentclass[preprint,12pt]{elsarticle}
\pdfoutput=1
\usepackage[ruled]{algorithm2e}

\SetAlFnt{\small}
\SetAlCapFnt{\small}
\SetAlCapNameFnt{\small}
\SetAlCapHSkip{0pt}
\IncMargin{-\parindent}

\usepackage{amsthm}
\usepackage{amsmath}
\usepackage{amssymb} 
\usepackage{graphicx}

\usepackage{mathrsfs}

\usepackage[T1]{fontenc}
\usepackage[utf8]{inputenc}
\usepackage{esint}
\usepackage{tikz}

\usepackage{wasysym}

\journal{ar\!Xiv.org for the record as a manuscript.
	}

\begin{document}

\newtheorem{theorem}{Theorem}[section]
\newtheorem{thm}[theorem]{Theorem}
\newtheorem{lem}[theorem]{Lemma}
\newtheorem{corollary}[theorem]{Corollary}
\newtheorem{cor}[theorem]{Corollary}
\newtheorem{df}[theorem]{Definition}
\newtheorem{ex}[theorem]{Example}
\newtheorem{mth}[theorem]{The Main Theorem}
\newtheorem{linthm}[theorem]{Linearity Theorem}
\newtheorem{defthm}[theorem]{Definability Theorem}
\newtheorem{mapth}[theorem]{The Mapping Theorem}
\newtheorem{mlm}[theorem]{The Main Lemma}
\newtheorem{remlm}[theorem]{The ${\tt RemoveMax} $ Invertibility Lemma}
\newtheorem{racrlm}[theorem]{The ${\tt RemoveAll} $ Cost Lemma}
\newtheorem{rmcrlm}[theorem]{The ${\tt RemoveMax} $ Cost Lemma}
\newtheorem{losslm}[theorem]{Credit Loss Characterization Lemma}
\newtheorem{pqlm}[theorem]{The pq Lemma}
\newtheorem{fixlm}[theorem]{Subheap Repair Lemma}
\newtheorem{unlm}[theorem]{The Uniqueness Lemma}
\newtheorem{complem}[theorem]{Competition Lemma}
\newtheorem{diaglem}[theorem]{Diagram Lemma}
\newtheorem{fundlem}[theorem]{Fundamental Lemma}
\newtheorem{charlem}[theorem]{Worst-case Heap Characterization Lemma}
\newtheorem{sLBlm}[theorem]{$ \sum \lambda $ Lower Bound Lemma}
\newtheorem{dth}[theorem]{Decomposition Theorem}
\newtheorem{sith}[theorem]{Singularity Theorem}
\newtheorem{1LBth}[theorem]{The 1$ ^{\mbox{st}} $ Lower Bound Theorem}
\newtheorem{2LBth}[theorem]{The Lower Bound Theorem}
\newtheorem{UBth}[theorem]{The Upper Bound Theorem}
\newtheorem{1oth}[theorem]{The 1$ ^{st} $ Optimality Theorem}
\newtheorem{2oth}[theorem]{The 2$ ^{nd} $ Optimality Theorem}
\newtheorem{3oth}[theorem]{The 3$ ^{rd} $ Optimality Theorem}
\newtheorem{cropt}[theorem]{Optimality Criterion}
\newtheorem{hyp}[theorem]{Hypothesis}
\newtheorem{example}[theorem]{Example}
\newtheorem{property}[theorem]{Property}
\newtheorem{note}[theorem]{Note}
\newtheorem{algMergeSort}[theorem]{Algorithm $ {\tt MergeSort} $}
\newtheorem{exercise}[theorem]{Exercise}

\newtheorem{theorem1}{Theorem}[subsection]
\newtheorem{thm1}[theorem1]{Theorem}
\newtheorem{lemma1}[theorem1]{Lemma}
\newtheorem{claim1}[theorem1]{Claim}
\newtheorem{corollary1}[theorem1]{Corollary}
\newtheorem{df1}[theorem1]{Definition}
\newtheorem{proposition1}[theorem1]{Proposition}
\newtheorem{problem1}[theorem1]{Problem}
\newtheorem{example1}{Example}[subsection]
\newtheorem{conjecture1}[theorem1]{Conjecture}
\newtheorem{remark1}[theorem1]{Remark}
\newtheorem{property1}[theorem1]{Property}
\newtheorem{algMergeSort1}[theorem1]{Algorithm $ {\tt MergeSort} $}
\newtheorem{linthm1}[theorem1]{Linearity Theorem}
\newtheorem{defthm1}[theorem1]{Definability Theorem}

\pagestyle{myheadings}
\markboth{M.A.Suchenek}{Suchenek:  Elementary Yet Precise Worst-case Analysis of MergeSort {\tt (SV)}}

\begin{frontmatter}

\title{Elementary Yet Precise Worst-case Analysis of MergeSort
\bigskip \\ {\tt \small A short version (SV) of a manuscript
intended for future publication}\tnoteref{label1} }
\tnotetext[label1]{\copyright 2017 Marek A. Suchenek.}
\author{MAREK A. SUCHENEK}

\address{California State University Dominguez Hills,
Department of Computer Science, \\
1000 E. Victoria St., Carson, CA 90747, USA,
  \texttt{Suchenek@csudh.edu}}

\begin{abstract}
The full version of this paper offers two elementary yet precise derivations of an exact formula 
 \[ W(n) = \sum_{i=1} ^{n} \lceil \lg i \rceil = n \lceil \lg n \rceil - 2^{\lceil \lg n \rceil} + 1 \]
for the maximum number $ W(n) $ of comparisons of keys performed by \linebreak $ {\tt MergeSort} $ on an $ n $-element array. The first of the two, due to its structural regularity, is well worth carefully studying in its own right. 

\smallskip

Close smooth bounds on $ W(n) $ are derived. It seems interesting that $ W(n) $ is linear between the points $ n = 2^{\lfloor \lg n \rfloor} $ and it linearly interpolates its own lower bound $ n \lg n - n + 1 $ between these points.

\smallskip

The manuscript (MS) of the full version of this paper, dated January 20, 2017, can be found at: \bigskip \\
\verb|http://csc.csudh.edu/suchenek/Papers/Analysis_of_MergeSort.pdf|

\bigskip

\end{abstract}


\end{frontmatter}

\tableofcontents

 \newpage

\section{Introduction}

$ {\tt MergeSort} $ is one of the fundamental sorting algorithms that is being taught in undergraduate Computer Science curricula across the U.S. and elsewhere. Its worst-case performance, measured by the number of comparisons of keys performed while sorting them, is optimal for the class of algorithms that sort inductively\footnote{Inductive sorting of $ n $ keys sorts a set of $ n-1 $ of those keys first, and then ``sorts-in'' the remaining $ n $-th key.} by comparisons of keys.\footnote{In its standard form analyzed in this paper, $ {\tt MergeSort} $ is not an inductive sorting algorithm. However, its worst-case performance, measured by the number of comparisons of keys performed while sorting them, is equal to the worst-case performance of the \textit{binary insertion sort} first described by Steinhaus in \cite{stein:matsnap} that is worst-case optimal in the class of inductive sorting algorithms that sort by comparisons of keys; see \cite{knu:art} page 186.} Historically, it\footnote{A bottom-up version of it, invented by John Neumann.} was the first sorting algorithm to run in $ O(n \lg n) $ time\footnote{In the worst case.}.

\medskip

So it seems only fitting to provide an exact formula for $ {\tt MergeSort} $'s worst-case performance \textit{and} derive it precisely.
Unfortunately, many otherwise decent texts offer  unnecessarily imprecise\footnote{Notable exceptions in this category are \cite{baavan:ana} and \cite{sedfla:ana} that derive almost exact formulas, but see Section~\ref{sec:oth} page~\pageref{sec:oth} for a brief critique of the results and their proofs offered there.} variants of it, and some with quite convoluted, incomplete, or incorrect proofs. Due to these imperfections, the fact that the worst-case performance of $ {\tt MergeSort} $
is the same as that of another benchmark sorting algorithm, the  \textit{binary insertion sort}  of \cite{stein:matsnap}, has remained unnoticed\footnote{Even in \cite{knu:art}.}.

\medskip

In this paper, I present two outlines\footnote{The detailed derivations can be found in \cite{suc:worstmergeMS}.} of elementary yet precise and complete derivations of an exact formula
 \[ W(n) = \sum_{i=1} ^{n} \lceil \lg i \rceil = n \lceil \lg n \rceil - 2^{\lceil \lg n \rceil} + 1 \]
for the maximum number $ W(n) $ \footnote{Elementary derivation of an exact formula for the \textit{best}-case performance $ B(n) $ of $ {\tt MergeSort} $, measured by the number of comparisons of keys performed while sorting them, has been done in \cite{suc:bestmerg}; see Section~\ref{sec:best} page~\pageref{sec:best} of this paper.} of comparisons of keys performed by \linebreak $ {\tt MergeSort} $ on an $ n $-element array. The first of the two, due to its structural regularity, is well worth carefully studying in its own right. 

\smallskip

Unlike some other basic sorting algorithms\footnote{For instance, $ {\tt Heapsort} $; see \cite{suc:wcheap} for a complete analysis of its worst-case behavior.} that run in $ O(n \lg n) $ time, $ {\tt MergeSort} $ exhibits a remarkably regular\footnote{As revealed by Theorem~\ref{thm:mersorbounds}, page~\pageref{thm:mersorbounds}.} worst-case behavior, 
the elegant simplicity of which
has been mostly lost on its rough analyses. In particular,
$ W(n) $ is linear\footnote{See Figure \ref{fig:boundsMergeSort} page \pageref{fig:boundsMergeSort}.} between the points $ n = 2^{\lfloor \lg n \rfloor} $ and it linearly interpolates its own lower bound $ n \lg n - n + 1 $ \footnote{Given by the left-hand side of the inequality~\eqref{eq:mersorbounds} page~\pageref{eq:mersorbounds}.} between these points.

\bigskip

What follows is a short version (SV) of a manuscript dated January 20, 2017, of the full version version \cite{suc:worstmergeMS} of this paper that has been posted at: \bigskip \\
\verb|http://csc.csudh.edu/suchenek/Papers/Analysis_of_MergeSort.pdf| \bigskip \\
The derivation of the worst case of $ {\tt MergeSort} $ presented here is roughly the same\footnote{Except for the present proof of Lemma~\ref{thm:sumCeilLog} which I haven't been using in my class.} as the one I have been doing in my undergraduate Analysis of Algorithms class. \ref{sec:slides} shows sample class notes from one of my lectures.


\section{Some Math prerequisites}

A manuscript of the full version \cite{suc:worstmergeMS} of this paper contains a clever derivation of a well-known\footnote{See \cite{knu:art}.} closed-form formula for
$ \sum_{i=1} ^{n} \lceil \lg i \rceil $. It proves insightful in my worst-case analysis of $ {\tt MergeSort} $ as its right-hand side will occur on page~\pageref{eq:result} in the fundamental equality~\eqref{eq:result} and serve as an instrument to derive the respective exact formula for $ {\tt MergeSort} $'s worst-case behavior.

\begin{lem} \label{thm:sumCeilLog}
For every integer $ n\geq 1 $,
\begin{equation} \label{eq:sumCeilLog} 
\sum_{i=1} ^{n} \lceil \lg i \rceil = \sum _{y=0}^{\lceil \lg n \rceil-1} ( n - 2^y) .
\end{equation}
\end{lem}
\noindent \textit{Proof} in \cite{suc:worstmergeMS}.
\hfill $ \Box $

\medskip

From this one can easily conclude that:

\begin{cor} \label{cor:sumCeilLog}
For every integer $ n\geq 1 $,
\begin{equation} \label{eq:sumCeilLog2} 
\sum_{i=1} ^{n} \lceil \lg i \rceil = n \lceil \lg n \rceil - 2^{\lceil \lg n \rceil} + 1 .
\end{equation}
\end{cor}

\section{$ {\tt MergeSort} $ and its worst-case behavior $ W(n) $} \label{sec:mergsor}

A call to $ {\tt MergeSort} $ inherits an $ n $-element array $ {\tt A} $ of integers and sorts it non-decreasingly, following the steps described below.

\begin{algMergeSort} \label{alg:mersor}
To 
sort an $ n $-element array $ {\tt A} $ do:

\begin{enumerate}
\item \label{alg:mersor:item1} If $ n \leq 1 $ then return $ {\tt A} $ to the caller,
\item \label{alg:mersor:item2} If $ n \geq 2 $ then
\begin{enumerate}
\item \label{alg:mersor:item2:it1} pass the first $ \lfloor \frac{n}{2} \rfloor $ elements of $ {\tt A} $ to a recursive call to $ {\tt MergeSort} $,
\item pass the last $ \lceil \frac{n}{2} \rceil $ elements of $ {\tt A} $ to another recursive call to $ {\tt MergeSort} $,
\item \label{alg:mersor:item2:it2} linearly merge, by means of a call to $ {\tt Merge} $, the non-decreasingly sorted arrays that were returned from those calls onto one non-decreasingly sorted array  $ {\tt A}^{\prime} $,
\item \label{alg:mersor:item2:it3} return $ {\tt A}^{\prime} $ to the caller.
 \end{enumerate}  
\end{enumerate}
\end{algMergeSort}

A Java code of 
$ {\tt Merge} $ is shown on the Figure~\ref{fig:Merge}.\footnote{A Java code of 
$ {\tt MergeSort} $ is shown in \ref{sec:appCode} Figure~\ref{fig:Sort} page~\pageref{fig:Sort}.}

\begin{figure}[h]
\centering
\includegraphics[width=0.7\linewidth]{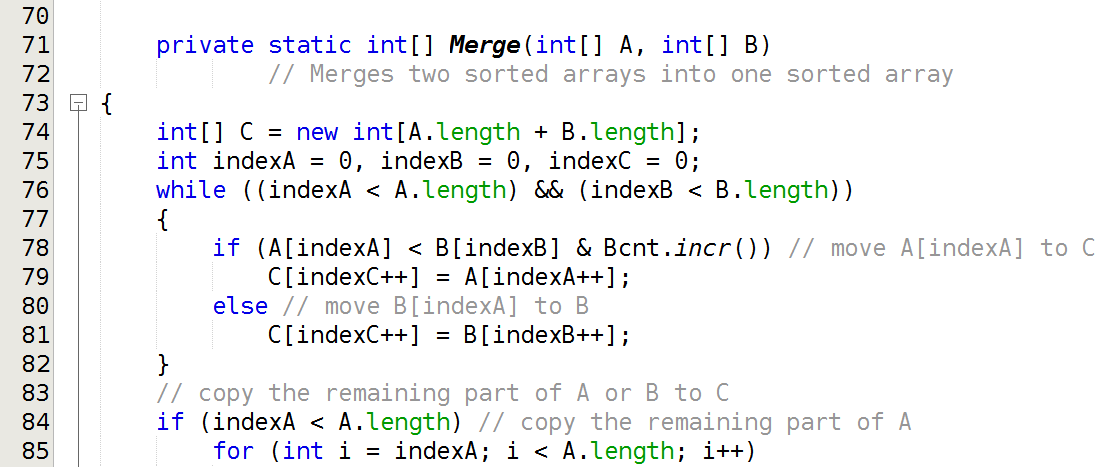} 
\includegraphics[width=0.7\linewidth]{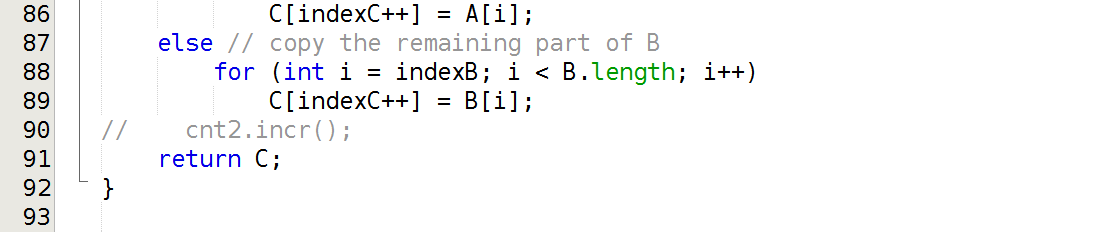}
\caption{A Java code of $ {\tt Merge} $, based on a pseudo-code from \cite{baa:ana}. Calls to $ {\tt Boolean} $ method $ {\tt Bcnt.incr()} $ count the number of comps for the purpose of experimental verification of the worst-case analysis of $ {\tt MergeSort} $.}
\label{fig:Merge}
\end{figure}

\medskip

A typical measure of the running time of $ {\tt MergeSort} $ is the number of \textit{comparisons of keys}, which for brevity I call \textit{comps}, that it performs while sorting array $ {\tt A} $. 

\begin{df}
The worst-case running time 
\[W(n)\]
of $ {\tt MergeSort} $ is defined as the maximum number of comps it performs while sorting an array of $ n $ distinct\footnote{This assumption is superfluous for the purpose of worst-case analysis as the mere presence of duplicates does not force $ {\tt MergeSort} $ to perform more comps.} elements. 
\end{df}

\label{pag:n>0}Clearly, if $ n=0 $ then $ W(n) = 0 $. From this point on, I am going to assume that $ n \geq 1 $.\footnote{This assumption turns out handy while using expression $ \lg n $.}

\medskip

Since no comps are performed outside $ {\tt Merge} $, $ W(n) $ can be computed as the sum of numbers of comps performed by all calls to $ {\tt Merge} $ during the execution of  $ {\tt MergeSort} $. The following classic results will be useful in my analysis.

\begin{theorem} \label{thm:merge_n-1}
The maximum number of comps performed by $ {\tt Merge} $ on two sorted list of total number $ n $ of elements is $ n-1 $.
\end{theorem}
\noindent \textit{Proof} (constructive, with Java code that generates worst cases shown in the \ref{sec:gen_worst_cases}) in \cite{suc:worstmergeMS}.
\hfill $ \Box $

\medskip

\label{pag:mergopt} Moreover, if the difference between the lengths of merged list is not larger than 1 then no algorithm that merges sorted lists by means of comps beats $ {\tt Merge} $ in the worst case, that is, has a lower than $ n-1 $ maximum number of comps.\footnote{Proof in \cite{knu:art}, Sec. 5.3.2 page 198; the worst-case optimality of $ {\tt Merge} $ ($ n-1 $ comps) was generalized in \cite{stoyao:optmerge} over lists of lengths $ k $ and $ m $, with $ k \leq m $, that satisfy $ 3k \geq 2m-2 $.}  This fact makes $ {\tt MergeSort} $ optimal in the intersection of the class of sorting algorithms that sort by merging two sorted lists of lengths' difference not larger than 1 \footnote{Or, by virtue of the above-quoted result from \cite{stoyao:optmerge}, with the difference not larger than the half of the length of the shorter list plus 1.}
with the class of sorting algorithms that sort by comps.

\section{An easy yet precise derivation of $ W(n) $} \label{sec:easy}

$ {\tt MergeSort} $  is a recursive algorithm. If $ n \geq 2 $ then it spurs a cascade of two or more recursive calls to itself.  A rudimentary analysis of the respective recursion tree $ T_n $, shown on Figure~\ref{fig:rectre}, yields a neat derivation of the exact formula for the maximum number $ W(n) $ of comps that $ {\tt MergeSort} $ performs on an $ n $-element array.

\begin{figure} [h]
\includegraphics[scale=.175]{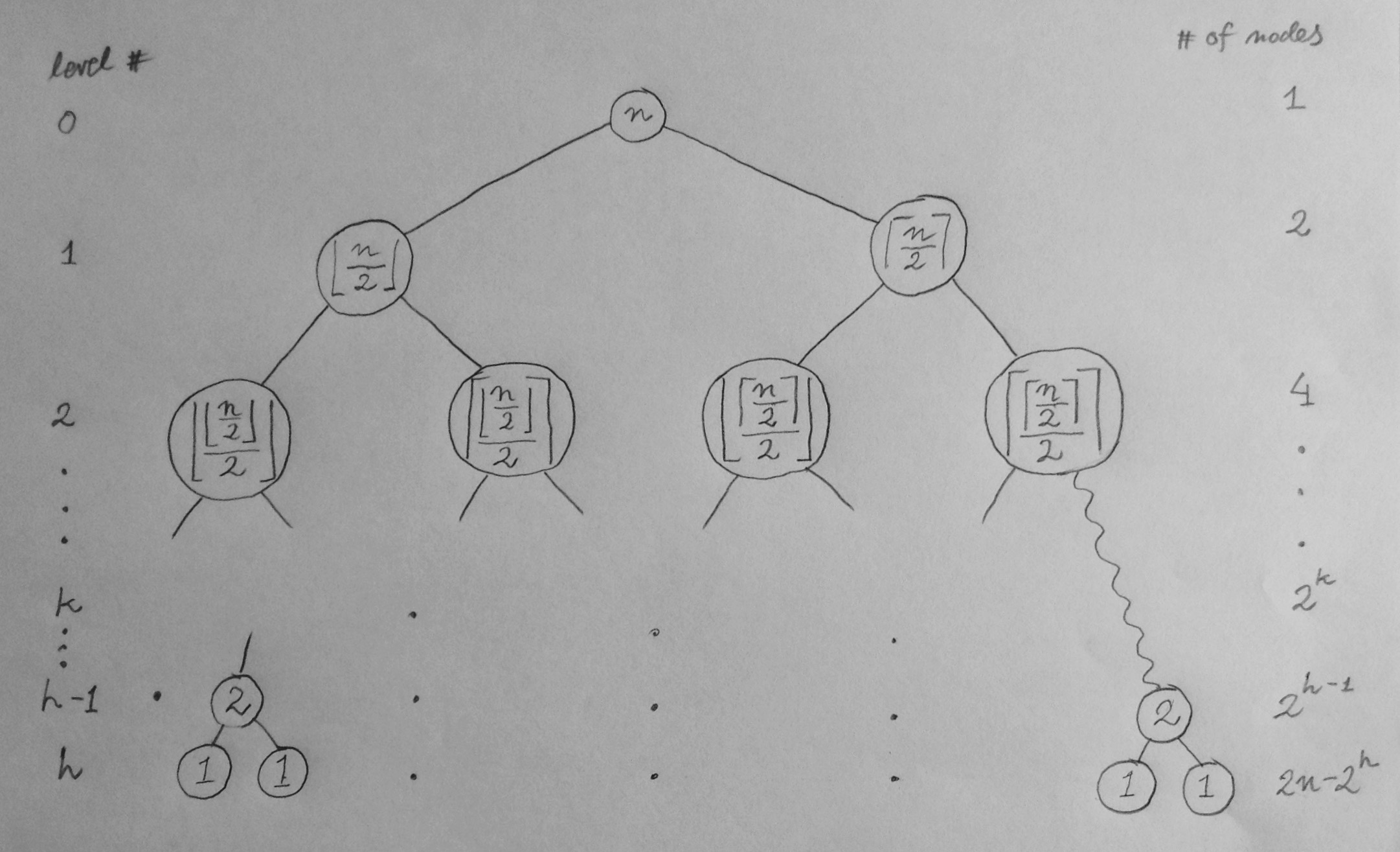} 
\caption{\label{fig:rectre} A sketch of the recursion 2-tree $ T_n $ for $ {\tt MergeSort} $ for a sufficiently large $ n $, with level numbers shown on the left and the numbers of nodes in the respective level shown on the right. The nodes correspond to calls to $ {\tt MergeSort} $ and show sizes of (sub)arrays passed to those calls. The last non-empty level is $ h $. The empty levels (all those numbered $ > h $) are not shown. The root corresponds to the original call to $ {\tt MergeSort} $. If a call that is represented by a node $ p $ executes further recursive calls to $ {\tt MergeSort} $ then these calls are represented by the children of $ p $; otherwise $ p $ is a leaf. The wavy line 
$ \photon $ 
represents a path in $ T_n $.}
\end{figure}

\medskip

The idea behind the derivation
is strikingly simple. It is based on the observation\footnote{Which I  prove in \cite{suc:worstmergeMS} as Theorem~4.6, page~14.} that for every  $ k \in \mathbb{N} $, the maximum number $ C_k $ of comps performed at each level\footnote{Empty or not.} $ k $  of $ T_n $ is given by this neat formula:\footnote{It is a simplification of formulas used in derivation presented in \cite{baavan:ana} and discussed in Section~\ref{sec:oth} page~\pageref{sec:oth}; in particular, it does not refer to the depth $ h $ of the decision tree $ T_n $.}
\begin{equation} \label{eq:main_level_comps_formula}
C_k = \max \{n - 2^k, 0\} .
\end{equation}
Since  
\begin{equation} \label{eq:hight_rec_tree}
n - 2^k > 0 \mbox{ if, and only if, }  \lceil \lg n \rceil - 1 \geq k ,
\end{equation}
the Corollary~\ref{cor:sumCeilLog} will allow me to conclude from \eqref{eq:main_level_comps_formula} and \eqref{eq:hight_rec_tree} the main result of this paper\footnote{This is how I have been deriving it in my undergraduate Analysis of Algorithms class for some 15 years or so, now.}:
\begin{equation} \label{eq:result}
W(n)  = \sum 
_{k \in \mathbb{N}}
 C_k = \sum _{k=0}^{\lceil \lg n \rceil-1} (n - 2^k) = n\lceil \lg n \rceil - 2^{\lceil \lg n \rceil} + 1 = \sum_{i=1} ^{n} \lceil \lg i \rceil.
\end{equation}

\bigskip 

The missing details\footnote{Which I did not show in my Analysis of Algorithms class.} in the above sketch are in \cite{suc:worstmergeMS}. Naturally, their only purpose is to prove the equality \eqref{eq:main_level_comps_formula} for all  $ k \in \mathbb{N} $,
as the rest, shown in \eqref{eq:result}, easily follows from it. 
In particular, we get:

\begin{mth} \label{thm:w-cMS}
The number $ W(n) $ of comparisons of keys that $ {\tt MergeSort} $ performs in the worst case while sorting an $ n $-element array is
\begin{equation} \label{eq:levComp2} 
W(n) = \sum_{i=1} ^{n} \lceil \lg i \rceil = n \lceil \lg n \rceil - 2^{\lceil \lg n \rceil} + 1 .
\end{equation}
\end{mth}
\noindent \textit{Proof} in \cite{suc:worstmergeMS}.
\hfill $ \Box $

\medskip

From that we can conclude a usual rough characterization of $ W(n) $: 
\[ W(n) \leq n (\lg n + 1) - 2^{ \lg n } + 1  =  n \lg n + n - n + 1 = n \lg n  + 1\]
and
\[ W(n) \geq n \lg n  - 2^{ \lg n + 1 } + 1  =  n \lg n - 2n + 1. \]
Therefore,
\[ W(n) \in \Theta (n \log n). \]

\medskip

The occurrence of $  \sum_{i=1} ^{n} \lceil \lg i \rceil $ in \eqref{eq:levComp2} allows to conclude that $ W(n) $ is exactly equal\footnote{\cite{knu:art} contains no mention of that fact.} to the number of comparisons of keys that the \textit{binary insertion sort}, considered by H. Steinhaus in \cite{stein:matsnap} and analyzed in \cite{knu:art}, performs in the worst case. Since the \textit{binary insertion sort} is known to be worst-case optimal\footnote{With respect to the number of comparisons of keys performed.} in the class of algorithms that perform incremental sorting, $ {\tt MergeSort} $ is worst-case optimal in that class\footnote{Although it is not a member of that class.}, too. From this and from the observation at the end of Section~\ref{sec:mergsor}, page~\pageref{pag:mergopt}, I conclude that no algorithm that sorts by merging two sorted lists and only by means of comps is worst-case optimal in the class of algorithms that sort by means of comps as it must perform 8 comps in the worst case while sorting 5 elements\footnote{They can be split in two: 1 plus 4, and follow the \textit{binary insertion sort}, or 2 plus 3, and follow $ {\tt MergeSort} $.}, while one can sort 5 elements by means of comps with no more than 7 comps.

\section{Close smooth bounds on $ W(n) $}

Our formula for $ W(n) $ contains a function ceiling that is harder to analyze than arithmetic functions and their inverses. In this Section, I outline a derivation of close lower and upper bounds on $ W(n) $ that are expressible by simple arithmetic formulas. I show that these bounds are the closest to $ W(n) $ in the class of functions of the form $  n \lg n + c n + 1 $, where $ c $ is a real constant.
\label{pag:epsilon} 
The detailed derivation and missing proofs can be found in \cite{suc:worstmergeMS}.

\medskip

Using the function $ \varepsilon $ (analyzed briefly in \cite{knu:art} and \cite{suc:wcheap}), a form of which is shown on Figure~\ref{fig:eps}, 
given by:
\begin{equation} \label{eq:eps}
\varepsilon  = 1 + \theta
- 2^{\theta} \mbox{ and } \theta =  \lceil \lg \, n \rceil -  \lg \, n ,
\end{equation}
one can conclude\footnote{\label{foo:proofdelta}See \cite{suc:wcheap}, Thm. 12.2 p. 94 for a proof.} that, for every $ n > 0 $,
\begin{equation} \label{eq:eps_basic_formula}
 n \lceil \lg n \rceil - 2^{\lceil \lg n \rceil} = n(\lg n + \varepsilon - 1) ,
\end{equation}
which yields
\begin{equation} \label{eq:W_eps}
 W(n) =  n( \lg n+ \varepsilon - 1) + 1 =  n \lg n + (\varepsilon - 1) n + 1 . 
\end{equation}
\begin{figure} [h]
\center
\includegraphics[scale=.34]{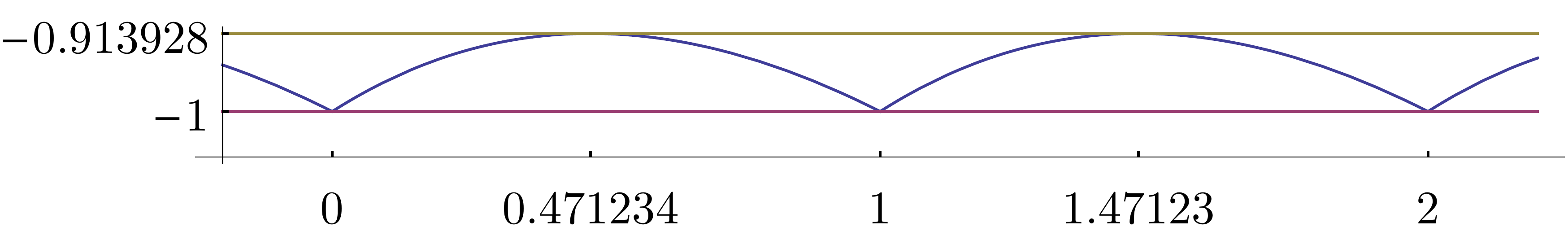}  
\caption{\label{fig:eps} Graph of $  \varepsilon - 1 $ as a function of $ \lg  n $.}
\end{figure}

\begin{property} \label{prop:delta}
Function $ \varepsilon $ given by \eqref{eq:eps} is a continuous function of $ n $ 
on the set of reals $ > 0 $. It assumes the minimum $ 0 $ for every  $ n = 2^{\lfloor \lg n \rfloor} $ and the maximum 
\begin{equation} \label{eq:Erdos}
\delta = 1 - \lg e + \lg \lg e \approx 0.0860713320559342 ,
\footnote{The constant $ 1 - \lg e + \lg \lg e $
has been known as the {\em Erd\"{o}s constant} $ \delta $. Erd\"{o}s used it around 1955 in order to establish an asymptotic upper bound for the number $ M(k) $ of different numbers in a multiplication table of size $ k \times k $ by means of the following limit:
\[\lim _{k \rightarrow \infty} \frac{\ln \frac{k \times k}{M(k)}}{\ln \ln (k \times k)} =  \delta .
\]}
\end{equation}
for every
\begin{equation} \label{eq:n_max_eps}
 n  =  2^{\lfloor \ln n + \lg\lg e \rfloor} \ln 2 
\end{equation}
and only such $ n $.
The function 
$ \varepsilon $ restricted to integers never reaches the value $\delta$. However,  $\delta$ is the  \textit{supremum} of $ \varepsilon $ restricted to integers. 
\end{property}
\noindent \textit{Proof} in \cite{suc:worstmergeMS}.
\hfill $ \Box $

Characterization \eqref{eq:W_eps} and Property~\ref{prop:delta} yield close smooth bounds of $ W(n) $. They are both of the form $  n \lg n + c n + 1 $ and they sandwich tightly $ W(n) $ between each other. If one sees $ W(n) $ as an infinite polygon\footnote{Which it is.}, its lower bound circumscribes it and its upper bound inscribes it.

\begin{theorem} \label{thm:mersorbounds} 
$ W(n) $ is a continuous concave function, linear between the points $ n =   2^{\lfloor \lg n  \rfloor}$, that for every $ n > 0 $ satisfies this inequality:
\begin{equation} \label{eq:mersorbounds} 
 n \lg n - n + 1 \leq W(n) \leq n \lg n - (1-\delta) n + 1 <  n \lg n - 0.913 n + 1 , 
\end{equation}
with the left $ \leq $ becoming $ = $ for every $ n = 2^{\lfloor \lg n \rfloor} $ and  the right $ \leq $ becoming $ = $ for every $ n = 2^{\lfloor \lg n + \lg \lg e \rfloor} \ln 2 $, and only for such $ n $. Moreover, the graph of $ W(n) $ is tangent to the graph of $ n \lg n - (1-\delta) n + 1 $ at the points $ n = 2^{\lfloor \lg n + \lg \lg e \rfloor} \ln 2 $, and only at such points.
\end{theorem}
\noindent \textit{Proof} in \cite{suc:worstmergeMS}.
\hfill $ \Box $

	\begin{figure} [h] \center
\includegraphics[scale=.34]{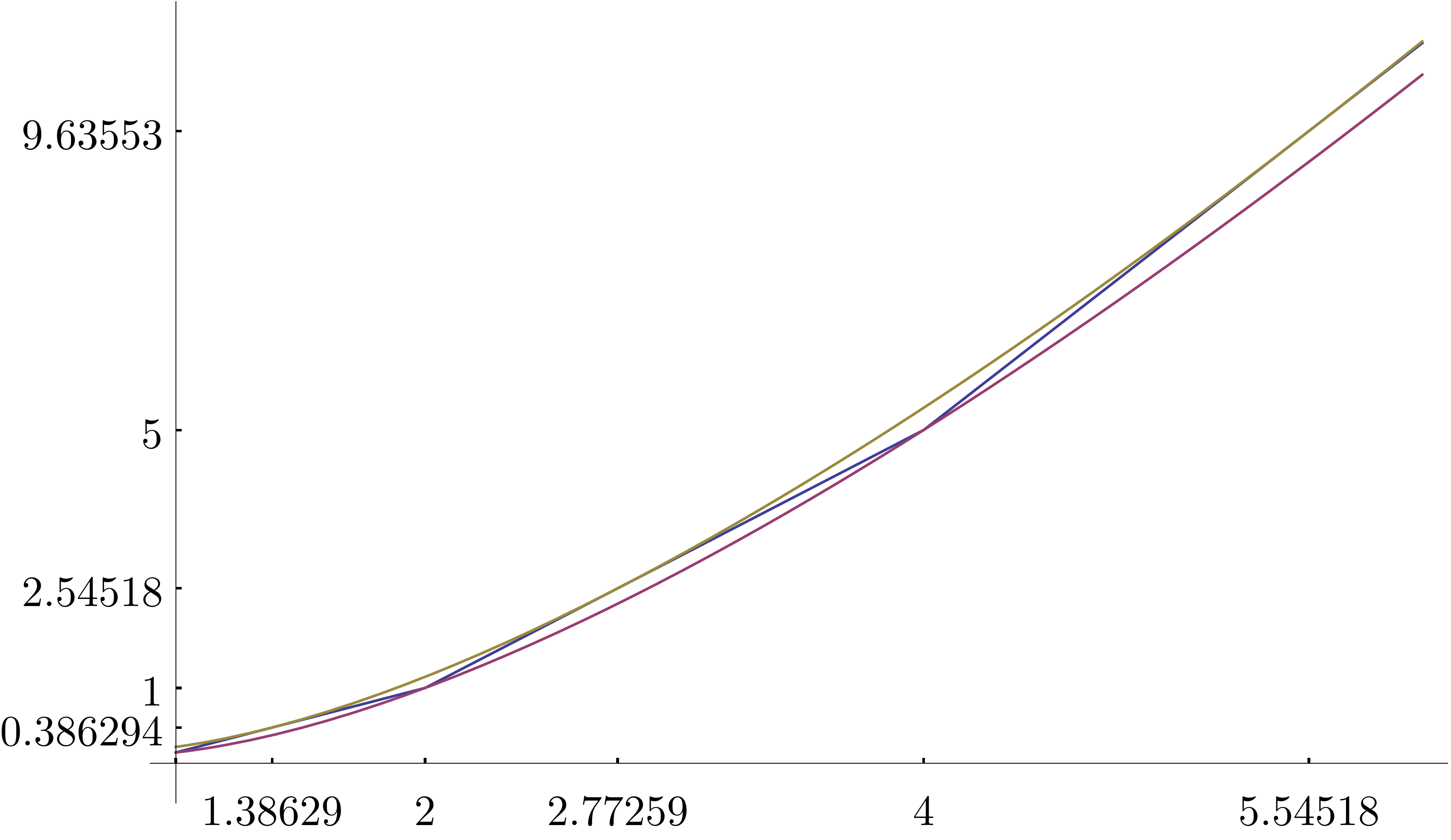}
		\caption{\label{fig:boundsMergeSort} $ W(n) = n\lceil \lg n \rceil - 2^{\lceil \lg n \rceil} +1 $ (the middle line) and its bounds $ n \lg n - n + 1 $ and $ n \lg n - (1-\delta) n + 1 $ $ \approx $ $ n \lg n - 0.913 n + 1 $, all three treated as functions of a positive real variable $ n $, plotted for $ n \in [1,6] $. $ W(n) $ is linear between the points $ n = 2^{\lfloor \lg n \rfloor} $ and it linearly interpolates its lower bound $ n \lg n - n + 1 $ between these points. Its upper bound $ n \lg n - (1-\delta) n + 1 $ inscribes it and is tangent to it at the points $  n $ $ = $  $ 2^{\lfloor \lg n + \lg\lg e \rfloor} \ln 2 $.}
	\end{figure}
	
	 \medskip

	 The bounds given by (\ref{eq:mersorbounds}) are really close\footnote{The distance between them is less than $  
	 \delta n \approx 0.0860713320559342 n $ for any positive integer $ n $.} to the exact value of $ W(n) $, as it is shown on Figure~\ref{fig:boundsMergeSort} page~\pageref{fig:boundsMergeSort}. The exact value $  n\lceil \lg n \rceil - 2^{\lceil \lg n \rceil} +1 $ is a continuous function (if $ n $ is interpreted as a real variable) despite that it incorporates discontinuous function \textit{ceiling}. 
	 
	\medskip
	
	\begin{note}
	\label{pag:interesting_interpolation}  It seems interesting that $ W(n) = n\lceil \lg n \rceil - 2^{\lceil \lg n \rceil} +1 $  (whether $ n $ is interpreted as a real variable or an integer variable) is linear between points $ n = 2^{\lfloor \lg n \rfloor} $ and linearly interpolates its own lower bound $ n \lg n - n + 1 $ between these points.
	\end{note}

	 \medskip
	
For $ n $ restricted to positive integers, the inequality \eqref{eq:mersorbounds} can be slightly enhanced by replacing the $ \leq $ symbol with $ < $, with the following result.
\begin{theorem} \label{thm:minconst}
$ 1-\delta $ is the greatest constant $ c $ such that for every integer $ n \geq 1 $, 
\begin{equation} \label{eq:minconst}
 W(n) < n \lg n - c n + 1.
\end{equation}
\end{theorem}
\noindent \textit{Proof} in \cite{suc:worstmergeMS}.
\hfill $ \Box $

\medskip

Theorem \ref{thm:minconst} can be reformulated as follows.

\begin{cor} \label{cor:minconstinf}
\begin{equation} \label{eq:minconstinf}
\inf \{c \in \mathbb{R} \mid \forall n \in \mathbb{N} \setminus \{0\} , W(n) < n \lg n - c n + 1 \}  = 1 - \delta .
\end{equation}
\end{cor}
\noindent \textit{Proof} in \cite{suc:worstmergeMS}.
\hfill $ \Box $

No upper bound of $ W(n) $ that has a form $ n \lg n - c n + 1 $ can coincide with $ W(n) $ at any integer $ n $, as the following fact ascertains.

\begin{cor} \label{cor:minconst_neq}
There is no constant $ c $ such that for every integer $ n \geq 1 $, 
\begin{equation} \label{eq:minconst_neq1}
 W(n) \leq n \lg n - c n + 1
\end{equation}
and for some integer $ n \geq 1 $, 
\begin{equation} \label{eq:minconst_neq2}
 W(n) = \lg n - c n + 1.
\end{equation}
\end{cor}
\noindent \textit{Proof} in \cite{suc:worstmergeMS}.
\hfill $ \Box $

\medskip

In particular\footnote{Note the $ \leq $ symbol in \eqref{eq:minconstinf2}.},
\begin{equation} \label{eq:minconstinf2}
\inf \{c \in \mathbb{R} \mid \forall n \in \mathbb{N} \setminus \{0\}, W(n) \leq n \lg n - c n + 1 \}  = 1 - \delta .
\end{equation}

\medskip

Moreover, we can conclude from Theorem \ref{thm:minconst} the following fact.

\begin{cor} \label{cor:minconst}
$ 1-\delta $ is the greatest constant $ c $ such that for every integer $ n \geq 1 $, 
\begin{equation} \label{eq:cor:minconst}
 W(n) \leq \lceil n \lg n - c n \rceil.
\end{equation}
\end{cor}
\noindent \textit{Proof} in \cite{suc:worstmergeMS}.
\hfill $ \Box $

\medskip

Since for any integer $ n \geq 1 $, $ W(n) $ is integer, the lower bound given by \eqref{eq:mersorbounds}
yields
\begin{equation} \label{eq:mersorbounds2} 
 \lceil n \lg n \rceil - n + 1 \leq W(n) \leq \lceil n \lg n - 0.913 n \rceil . 
\end{equation}
By virtue of Corollary \ref{cor:minconst}, for some integers $ n \geq 1 $,\footnote{For instance, for $ n=11 $.}
\begin{equation} \label{eq:wrongbound}
W(n) > \lceil n \lg n - 0.914 n \rceil .
\end{equation}
Although the bounds given by \eqref{eq:mersorbounds2} \footnote{Almost the same bounds were given in \cite{baavan:ana}; see Section~\ref{sec:oth} for more details on this.} are tighter than those given by \eqref{eq:mersorbounds}, they nevertheless involve the discontinuous \textit{ceiling} function, so that they may not be as easy to visualize or analyze as some differentiable functions, thus losing their advantage over the precise formula $ W(n) = n\lceil \lg n \rceil - 2^{\lceil \lg n \rceil} +1 $. Therefore, the bounds given by \eqref{eq:mersorbounds} appear to have an analytic advantage over those given by \eqref{eq:mersorbounds2}.

\section{Other properties of the recursion tree $ T_n $}

This sections contains some well-known auxiliary facts that I didn't need for the derivation of the exact formula for $ W(n) $ but am going to derive from the Main Lemma 4.1 of \cite{suc:worstmergeMS} for the sake of a thoroughness of my analysis of the decision tree $ T_n $.

\begin{thm} \label{thm:depthRectree}
The depth $ h $ of the recursion tree $ T_n $ is 
\begin{equation} \label{eq:height}
h = \lceil \lg n \rceil .
\end{equation}
\end{thm}
\noindent \textit{Proof} in \cite{suc:worstmergeMS}.
\hfill $ \Box $
\medskip

\begin{note}
Theorem~\ref{thm:depthRectree} allows for quick derivation of fairly close upper bound on the number of comps performed by $ {\tt MergeSort} $ on an $ n $-element array. Since at each level of $ T_n $ less than $ n $ comparisons are performed by $ {\tt Merge} $ and at level $ h $ no comps are performed, and there are $ h $ $ = $ $ \lceil \lg n \rceil  $ levels below level $ h $, the total number of comps is not larger than
\begin{equation} \label{eq:UBapprox} 
(n-1)h = (n-1)(\lceil \lg n \rceil) < (n-1) (\lg n + 1)   \in O(n \log n).
\end{equation}
\end{note}

 \medskip

\medskip

A \textit{cut} of a tree $ T_n $ is a set $ \Gamma $ of nodes of $ T $ such that every branch\footnote{A maximal path.} in $ T_n $ has exactly one element in $ \Gamma $.

\begin{thm} \label{thm:cut}
The sum of values shown at the elements of any cut of $ T_n $ is $ n $.
\end{thm}
\noindent \textit{Proof} in \cite{suc:worstmergeMS}.
\hfill $ \Box $

\begin{thm} \label{thm:sum_leaves}
The number of leaves in the recursion tree $ T_n $ is $ n $.
\end{thm}
\noindent \textit{Proof} in \cite{suc:worstmergeMS}.
\hfill $ \Box $

The following corollary provides some statistics about recursive calls to $ {\tt MergeSort} $.

\begin{cor} \label{note:calls_stats}
For every integer $ n > 0 $,
\begin{enumerate}
     \renewcommand\labelenumi{\theenumi}
     \renewcommand{\theenumi}{(\roman{enumi})}
\item \label{note:calls_stats:1} $ T_n $ has $ 2n - 1 $ nodes.

\item  \label{note:calls_stats:2} The number or recursive calls spurred by $ {\tt MergeSort} $ on any $ n $-element array is $ 2(n - 1) $.

\item  \label{note:calls_stats:3}
The sum $ S_n $ of all values shown in the recursion tree $ T_n $ on Figure~\ref{fig:rectre} is equal to:
\begin{equation} \label{eq:sum_nodes}
S_n =  n \lceil \lg n \rceil - 2^{\lceil \lg n \rceil} + 2n  = n(\lg n + \varepsilon + 1).
\end{equation}
\item \label{note:calls_stats:4}
The average size $ A_n $ of array passed to any recursive call to $ {\tt MergeSort} $ while sorting an $ n $-element array is:
\begin{equation} \label{eq:avg_nodes}
A_n =  \frac{1}{2}(1+\frac{1}{n-1}) (\lg n + \varepsilon) \approx \frac{1}{2} (\lg n + \varepsilon).
\end{equation}
\end{enumerate}
\end{cor}
\noindent \textit{Proof} in \cite{suc:worstmergeMS}.
\hfill $ \Box $

\medskip

 Here is a very insightful property. It states that $ {\tt MergeSort} $ is splitting its input array fairly evenly\footnote{The sizes of the sub-arrays passed to recursive calls at any non-empty level $ k $ of the decision tree $ T_n $ above the last non-empty level $ h $ are the same as the sizes of the elements of the maximally even partition of an $ n $-element set onto $ 2^k $ subsets.} so that at any level of the recursive tree, the difference between the lengths of the longest sub-array and the shortest sub-array is $ \leq 1. $ This fact is the root cause of 
good worst-case performance
of $ {\tt MergeSort} $.

\begin{property} \label{pro:level} 
The difference between values shown by any two nodes in the same level of of the recursion tree $ T_n $ is $ \leq 1 $.
\end{property} 
\noindent \textit{Proof} in \cite{suc:worstmergeMS}.
\hfill $ \Box $

\label{pag:merg_opti_for_mersor}Property \ref{pro:level} has this important consequence that $ {\tt Merge} $ is, by virtue of the observation on page~\pageref{pag:mergopt} after the Theorem~\ref{thm:merge_n-1} page~\pageref{thm:merge_n-1}, worst-case comparison-optimal while merging any two sub-arrays of the same level of the recursion tree. Thus the worst-case of $ {\tt MergeSort} $ cannot be improved just by replacing $ {\tt Merge} $ with some tricky merging $ {\tt X} $ as long as $ {\tt X} $ merges by means of comparisons of keys.

\begin{cor} \label{cor:merg_opti_for_mersor}
Replacing  {\tt Merge} with any other method that merges sorted arrays by means of comps will not improve the worst-case performance of $ {\tt MergeSort} $ measured with the number of comps while sorting an array.
\end{cor}
\begin{proof}
Proof follows from the  above observation.
\end{proof}

\medskip

Since a parent must show a larger value than any of its children, the Property~\ref{pro:level} has also the following consequence.
\begin{cor} \label{cor:last2}
The leaves in the recursion tree $ T_n $ can only reside at the last two non-empty levels of $ T_n $. 
\end{cor}
\begin{proof}
Proof follows from the Property~\ref{pro:level} as the above observation indicates.
\end{proof}

As a result, one can conclude\footnote{Cf. \cite{knu:art}, Sec. 5.3.1 Ex. 20 page 195.} that the recursion tree $ T_n $ has the miminum \textit{internal} and \textit{external path length}s among all binary trees on $ 2n - 1 $ nodes.

\medskip

Since all nodes at the level $ h $ of the recursion tree $ T_n $ are leaves and show value $ 1 $, no node at level $ h-1 $ can show a value $ >2 $. Indeed, level $ h-1 $ may only contain leaves, that show value $ 1 $, and parents of nodes of level $ h $ that show value $ 1+1=2 $. This observation and the previous result allow for easy characterization of contents of the last two non-empty levels of tree $ T_n $.


\noindent \begin{cor} \label{cor:last2char}
For every $ n\geq 2  $: 

\begin{enumerate}
     \renewcommand\labelenumi{\theenumi}
     \renewcommand{\theenumi}{(\roman{enumi})}
\item \label{cor:last2char:item1} there are $ 2^h - n $ leaves, all showing value $ 1 $, 
at the level $ h-1 $,
\item \label{cor:last2char:item2} there are $  n - 2^{h-1} $ non-leaves, all showing value $ 2 $, 
at the level $ h-1 $,
and
\item \label{cor:last2char:item3} there are $  2n - 2^{h} $ \footnote{This value shows in the lower right corner of Figure~\ref{fig:rectre} page~\pageref{fig:rectre} of a sketch of the recursion tree $ T_n $; it was not need needed for the derivation of the main result \eqref{eq:levComp2} page~\pageref{eq:levComp2}, included for the sake of completeness only.} nodes, all leaves showing value $ 1 $, at the level $ h $
\end{enumerate}
\noindent of the recursion tree $ T_n $, where $ h $ is the depth\footnote{The level number of the last non-empty level of $ T_n $.} of $ T_n $.
\end{cor}
\noindent \textit{Proof} in \cite{suc:worstmergeMS}.
\hfill $ \Box $

\section{A derivation of $ W(n) $ without references to the recursion tree}

In order to formally prove Theorem~\ref{thm:w-cMS} without any reference to the recursion tree, I use here the well-known\footnote{For instance, derived in \cite{baa:ana} and \cite{baavan:ana}.} recurrence relation
\begin{equation} \label{eq:recmergesort1} 
W(n) = W(\lfloor \frac{n}{2} \rfloor) + W(\lceil \frac{n}{2} \rceil) + n - 1 \mbox{ if } n \geq 2
\end{equation}
\begin{equation} \label{eq:recmergesort2} 
W(1) = 0
\end{equation}
that easily follows from the description (Algorithm~\ref{alg:mersor} page~\pageref{alg:mersor}) of $ {\tt MergeSort} $, steps~\ref{alg:mersor:item2:it1}, \ref{alg:mersor:item2:it2} and Theorem~\ref{thm:merge_n-1}. I am going to prove, by direct inspection,  that the function $ W(n) $ defined by (\ref{eq:levComp2}) satisfies equations (\ref{eq:recmergesort1}) and (\ref{eq:recmergesort2}).

\medskip

The details of the proof are in \cite{suc:worstmergeMS}.

\section{Other work} \label{sec:oth}

Although some variants of parts of the formula \eqref{eq:levComp2} appear to have been known for quite some time now, even otherwise precise texts offer derivations that leave room for improvement. For instance, the recurrence relation for $ {\tt MergeSort} $ analyzed in \cite{sedfla:ana} asserts that the least number of comparisons of keys performed outside the recursive calls, if any, that suffice to sort an array of size $ n $ is $ n $ rather than $ n-1 $.  This seemingly inconsequential variation results in a solution $ W(n) =  \sum_{i=1} ^{n-1} (\lfloor \lg i \rfloor + 2)$ \footnote{I saw $ W(n) =  \sum_{i=1} ^{n-1} \lfloor \lg i \rfloor  $ on slides that accompany \cite{sedfla:ana}.} on page 2, Exercise 1.4, rather than the correct formula (\ref{eq:result})  $ W(n) =  \sum_{i=1} ^{n} \lceil \lg i \rceil $ derived in this paper. (Also, the relevant derivations presented in  \cite{sedfla:ana}, although quite clever, are not nearly as precise and elementary as those presented in this paper.) As a result, the fact that $ {\tt MergeSort} $ performs exactly the same number of comparisons of keys as does another classic, \textit{binary insertion sort}, considered by H. Steinhaus and analyzed in \cite{knu:art}, remains unnoticed.

\bigskip

\bigskip

Pages 176 -- 177 of \cite{baavan:ana} contain 
an early sketch of proof of 
\begin{equation} \label{eq:baavan_result}
W(n) = n h - 2^{h} + 1,
\end{equation}
where $ h $ is the depth of the recursion tree $ T_n $, with remarkably close\footnote{Although not 100 percent correct.} bounds given by~\eqref{eq:baavan_bound} page~\pageref{eq:baavan_bound}. It is similar\footnote{The idea behind the sketch of the derivation in \cite{baavan:ana} was based on an observation that \[W(n) = \sum _{i = 0} ^{h-2} (n - 2^i) + \frac{n - B}{2} , \] where $ B $ was the number of leaves at the level $ h-1 $ of the decision tree $ T_n $; it was sketchily derived from the recursion tree shown on Figure~\ref{fig:RecTreeBaaVan} and properties stated in the Corollary~\ref{cor:last2char} page~\pageref{cor:last2char} (with only a sketch of proof in \cite{baavan:ana})  not needed for the derivation presented in Section~\ref{sec:easy}.}  to a simpler 
derivation based on the equality \eqref{eq:main_level_comps_formula}, presented in this paper in Section~\ref{sec:easy} and outlined in~\eqref{eq:result} page~\pageref{eq:result} (except for the $ \sum_{i=1} ^{n} \lceil \lg i \rceil $ part), which it predates by several years. 

\medskip

\begin{figure}[h]
\centering
\includegraphics[width=1\linewidth]{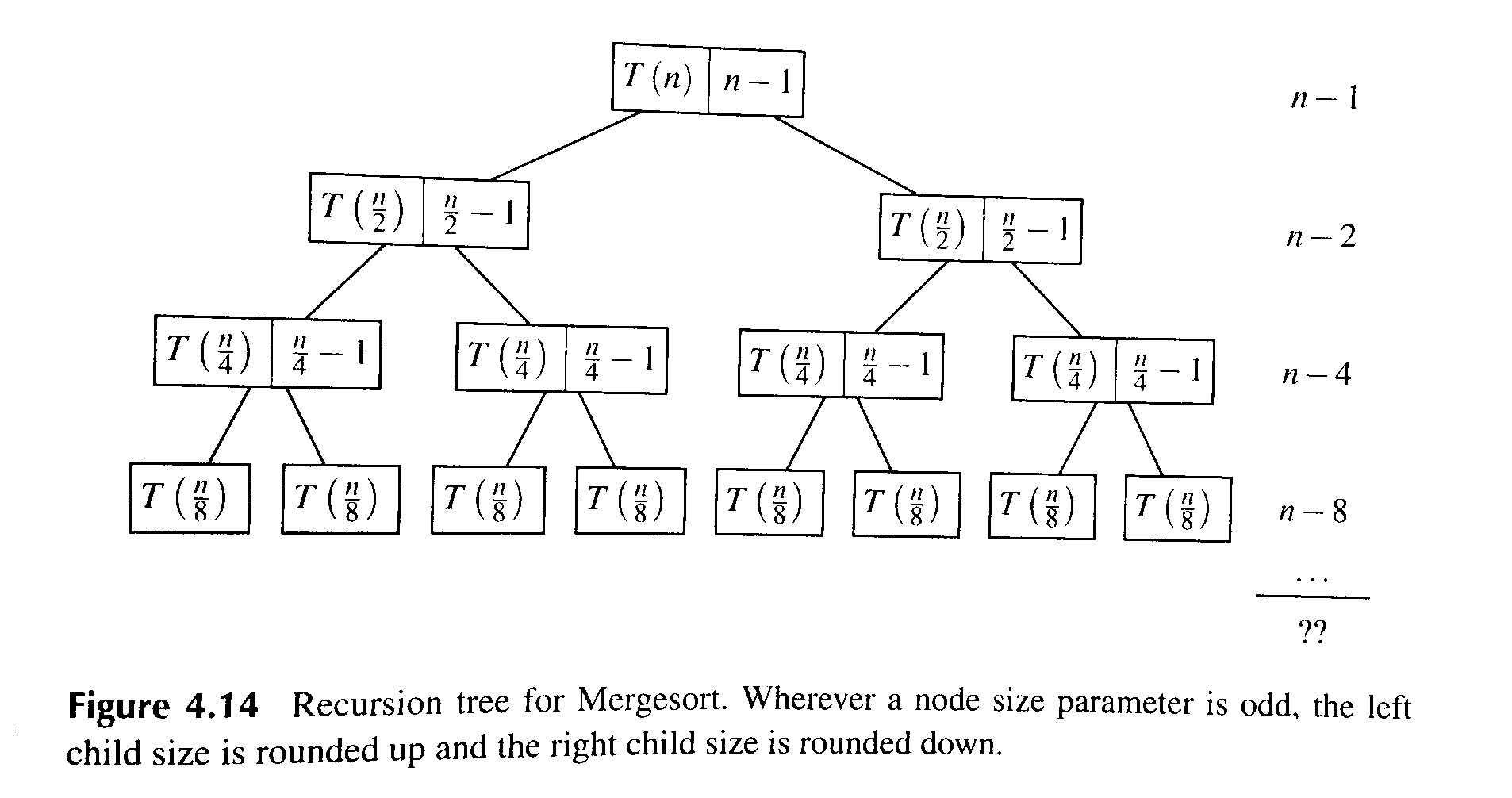}
\caption{A snapshot from \cite{baavan:ana}, page 177, showing a decision tree for $ {\tt MergeSort}$. \textit{\textbf{Note}}: This picture is copyrighted by \textit{Addison Wesley Longman} (2000). It was reproduced here from \cite{baavan:ana} for \textit{criticism} and \textit{comment} purposes only, and not for any other purpose, as prescribed by U.S. Code Tittle 17 Chapter 1 para 107 that established the ``fair use'' exception of copyrighted material. }
\label{fig:RecTreeBaaVan}
\end{figure}

The \cite{baavan:ana}'s version of the decision tree $ T_n $ (Figure 4.14 page 177 of \cite{baavan:ana}, shown here on Figure~\ref{fig:RecTreeBaaVan}) was a re-use of a decision tree for the special case of $ n = 2^{\lfloor \lg n \rfloor} $, with an ambiguous, if at all correct\footnote{It may be interpreted as to imply that for any level $ k $, all the left-child sizes at level $ k $ are the same and all the right-child sizes at level $ k $ are the same, neither of which is a valid statement.}, comment in the caption that ``[w]henever a node size parameter is odd, the left child size parameter is rounded up\footnote{Should be: \textit{down}, according to \eqref{eq:recmergesort1} page~\pageref{eq:recmergesort1}.} and the right child size is rounded down\footnote{Should be: \textit{up}, according to \eqref{eq:recmergesort1} page~\pageref{eq:recmergesort1}.}.'' The proof of the fact, needed for the derivation in \cite{baavan:ana}, that $ T_n $ had no leaves outside its last two levels (Corollary~\ref{cor:last2} page~\pageref{cor:last2}, not needed for the derivation presented in Section~\ref{sec:easy}) was waved with a claim ``[w]e can\footnote{This I do not doubt.} determine that [...]''

\medskip

Although $ h $ was claimed in \cite{baavan:ana} to be equal to $ \lceil \lg (n+1) \rceil $~\footnote{Which claim must have produced an incorrect formula $ n \lceil \lg (n+1) \rceil - 2^{\lceil \lg (n+1) \rceil} + 1 $ for $ W(n) $ and precluded concluding the neat characterization $ W(n) = \sum_{i=1} ^{n} \lceil \lg i \rceil $.} (and not to the correct $ \lceil \lg n \rceil $ given by the equality~\eqref{eq:height} page~\pageref{eq:height}, a fact not needed for the derivation presented in Section~\ref{sec:easy}), somehow the mostly correct conclusion\footnote{Almost identical with \eqref{eq:mersorbounds2} page~\pageref{eq:mersorbounds2}, except for the constant $ 0.914 $.}
was inferred from it, however, with no details offered - except for a mention that a function $ \alpha $ that satisfies $ h = \lg n + \lg \alpha $, similar to function $ \varepsilon $ shown on Figure~\ref{fig:eps} page~\pageref{fig:eps}, was used. It stated that (Theorem 4.6, page 177, in \cite{baavan:ana}):
\begin{equation} \label{eq:baavan_bound}
 \lceil n \lg n - n + 1 \rceil \leq W(n) \leq \lceil n \lg n - 0.914 n \rceil .
\end{equation}
		
It follows from \eqref{eq:wrongbound} page \pageref{eq:wrongbound} that the constant $ 0.914 $ that appears in \eqref{eq:baavan_bound} is incorrect. It was a rounding error\footnote{Of  $ 1 - \delta $, where $ \delta $ is given by \eqref{eq:Erdos} page~\pageref{eq:Erdos}.}, I suppose, that produced a false upper bound\footnote{For instance, if $ n = 11 $ then $ {\tt MergeSort} $ performs 29 comparisons of keys while the value of the upper bound $ \lceil n \lg n - .914 n \rceil $ given in \cite{baavan:ana}, Theorem 4.6. p. 177, is 28; this is a significant error as 28 or less comps while sorting any 11-element array beats the \textit{binary insertion sort} that requires $ \sum_{i=1} ^{11} \lceil \lg i \rceil = 29$ comps in the worst case.}.   

\section{Best-case analysis of $ {\tt MergeSort} $} \label{sec:best}

It turns out that derivation of minimum number $ B(n) $ of comps performed by $ {\tt MergeSort} $ on an $ n $-element array is a bit more tricky. A formula
\begin{equation} \label{eq:best}
 \frac{n}{2}(\lfloor \lg n \rfloor + 1) -  \sum _{k=0} ^{\lfloor \lg n \rfloor} 2^k  \mbox{\textit{Zigzag}}\,(\frac{n}{2^{k+1}}),
\end{equation}
 where
 \[ \mbox{\textit{Zigzag}} (x)  =  \min (x - \lfloor x \rfloor, \lceil x \rceil - x), \]
 has been derived and thoroughly analyzed in \cite{suc:bestmerg}. It has been also demonstrated in \cite{suc:bestmerg} that there is no closed-form formula for $ B(n) $.
 
 \medskip
 
 Incidentally, as it was pointed out in \cite{suc:bestmerg} $, B(n) $ is equal to the sum $ A(n,2) $ of bits in binary representations of all integers $ < n $.

\newpage

\appendix

\section{A Java code of $ {\tt MergeSort} $} \label{sec:appCode}

Figure~\ref{fig:Sort} shows a Java code of $ {\tt MergeSort} $.

\begin{figure}[h]
\centering
\includegraphics[width=.95\linewidth]{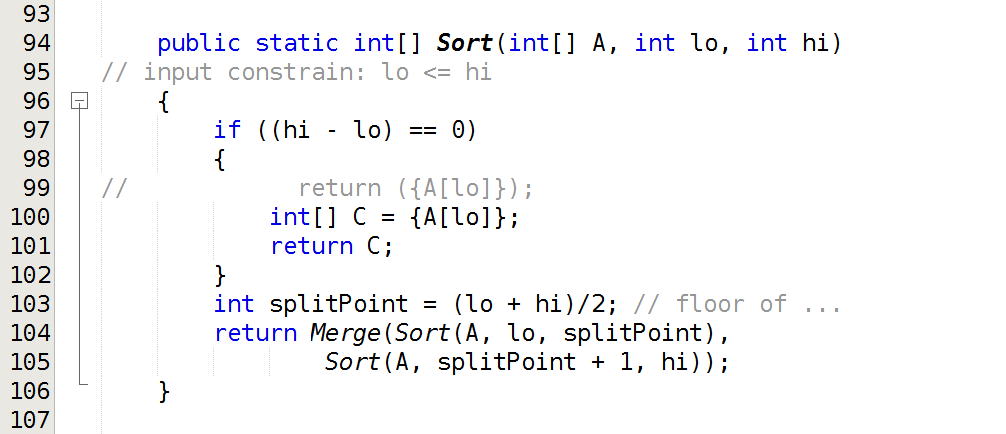}
\caption{A Java code of $ {\tt MergeSort} $. A code of $ {\tt Merge} $ is shown on Figure~\ref{fig:Merge}.}
\label{fig:Sort}
\end{figure}

\section{Generating worst-case arrays for $ {\tt MergeSort} $} \label{sec:gen_worst_cases}

Figure~\ref{fig:wostGenMerge} shows a self-explanatory Java code of recursive method $ \tt unSort $ that given a sorted array $ {\tt A} $ reshuffles it, in a way resembling $\tt InsertionSort $\footnote{Although not with $\tt InsertionSort $'s sluggishness; the number of moves of keys it performs is only slightly more than the \textit{minimum} number \eqref{eq:best} of comps performed by $ {\tt MergeSort} $ on any $ n $-element array.}, onto a worst-case array for $ {\tt MergeSort} $.

\begin{figure}[h]
\centering
\includegraphics[width=.9\linewidth]{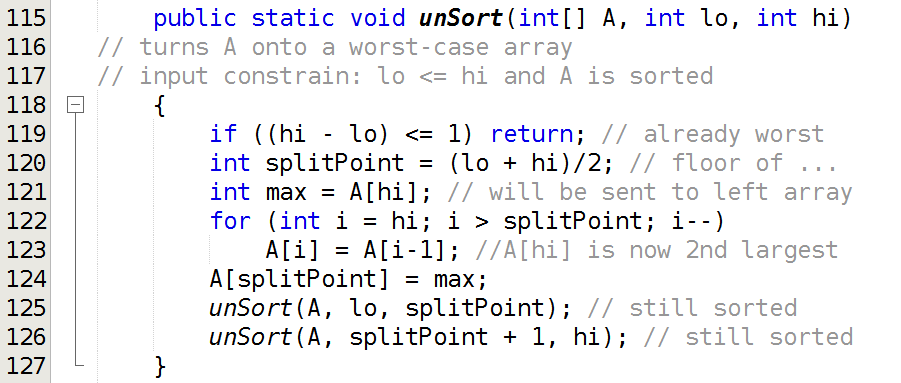}
\caption{A Java code of $ {\tt unSort} $ that, given a sorted array $ {\tt A} $, reshuffles it onto a worst-case array for $ {\tt MergeSort} $. Its structure mimics the Java code of $ {\tt MergeSort} $ shown on Figure~\ref{fig:Sort}.}
\label{fig:wostGenMerge}
\end{figure}

\medskip

For instance, it produced this array of integers between 1 and 500:

\medskip

\noindent 1, 500, 2, 3, 4, 7, 5, 6, 8, 15, 9, 10, 11, 14, 12, 13, 16, 31, 17, 18, 19, 22, 20, 21, 23, 30, 24, 25, 26, 29, 27, 28, 32, 62, 33, 34, 35, 38, 36, 37, 39, 46, 40, 41, 42, 45, 43, 44, 47, 61, 48, 49, 50, 53, 51, 52, 54, 60, 55, 56, 57, 59, 58, 63, 124, 64, 65, 66, 69, 67, 68, 70, 77, 71, 72, 73, 76, 74, 75, 78, 92, 79, 80, 81, 84, 82, 83, 85, 91, 86, 87, 88, 90, 89, 93, 123, 94, 95, 96, 99, 97, 98, 100, 107, 101, 102, 103, 106, 104, 105, 108, 122, 109, 110, 111, 114, 112, 113, 115, 121, 116, 117, 118, 120, 119, 125, 249, 126, 127, 128, 131, 129, 130, 132, 139, 133, 134, 135, 138, 136, 137, 140, 155, 141, 142, 143, 146, 144, 145, 147, 154, 148, 149, 150, 153, 151, 152, 156, 186, 157, 158, 159, 162, 160, 161, 163, 170, 164, 165, 166, 169, 167, 168, 171, 185, 172, 173, 174, 177, 175, 176, 178, 184, 179, 180, 181, 183, 182, 187, 248, 188, 189, 190, 193, 191, 192, 194, 201, 195, 196, 197, 200, 198, 199, 202, 216, 203, 204, 205, 208, 206, 207, 209, 215, 210, 211, 212, 214, 213, 217, 247, 218, 219, 220, 223, 221, 222, 224, 231, 225, 226, 227, 230, 228, 229, 232, 246, 233, 234, 235, 238, 236, 237, 239, 245, 240, 241, 242, 244, 243, 250, 499, 251, 252, 253, 256, 254, 255, 257, 264, 258, 259, 260, 263, 261, 262, 265, 280, 266, 267, 268, 271, 269, 270, 272, 279, 273, 274, 275, 278, 276, 277, 281, 311, 282, 283, 284, 287, 285, 286, 288, 295, 289, 290, 291, 294, 292, 293, 296, 310, 297, 298, 299, 302, 300, 301, 303, 309, 304, 305, 306, 308, 307, 312, 373, 313, 314, 315, 318, 316, 317, 319, 326, 320, 321, 322, 325, 323, 324, 327, 341, 328, 329, 330, 333, 331, 332, 334, 340, 335, 336, 337, 339, 338, 342, 372, 343, 344, 345, 348, 346, 347, 349, 356, 350, 351, 352, 355, 353, 354, 357, 371, 358, 359, 360, 363, 361, 362, 364, 370, 365, 366, 367, 369, 368, 374, 498, 375, 376, 377, 380, 378, 379, 381, 388, 382, 383, 384, 387, 385, 386, 389, 404, 390, 391, 392, 395, 393, 394, 396, 403, 397, 398, 399, 402, 400, 401, 405, 435, 406, 407, 408, 411, 409, 410, 412, 419, 413, 414, 415, 418, 416, 417, 420, 434, 421, 422, 423, 426, 424, 425, 427, 433, 428, 429, 430, 432, 431, 436, 497, 437, 438, 439, 442, 440, 441, 443, 450, 444, 445, 446, 449, 447, 448, 451, 465, 452, 453, 454, 457, 455, 456, 458, 464, 459, 460, 461, 463, 462, 466, 496, 467, 468, 469, 472, 470, 471, 473, 480, 474, 475, 476, 479, 477, 478, 481, 495, 482, 483, 484, 487, 485, 486, 488, 494, 489, 490, 491, 493, 492. 

\medskip

It took my $ {\tt MergeSort} $  3,989 comps to sort it. Of course,
\[500 \lceil \lg 500 \rceil - 2^{\lceil \lg 500 \rceil} + 1 = 4,500 - 512 + 1 = 3,989 .\]

\section{Notes from my Analysis of Algorithms lecture} \label{sec:slides}

Below are some of the class digital notes I wrote while lecturing Analysis of Algorithms in Spring 2012, with some comments added after class. Figure~4.14 (decision tree) is from the course textbook \cite{baavan:ana}, page 177, showing a decision tree for $ {\tt MergeSort}$. \textit{\textbf{Note}}: This figure is copyrighted by \textit{Addison Wesley Longman} (2000). I used it \textit{transformatively} in my class for nonprofit \textit{education}, \textit{criticism}, and \textit{comment}  purposes only, and not for any other purpose, as prescribed by U.S. Code Tittle 17 Chapter 1 para 107 that established the ``fair use'' exception of copyrighted material.

\bigskip

\includegraphics[width=0.9\linewidth]{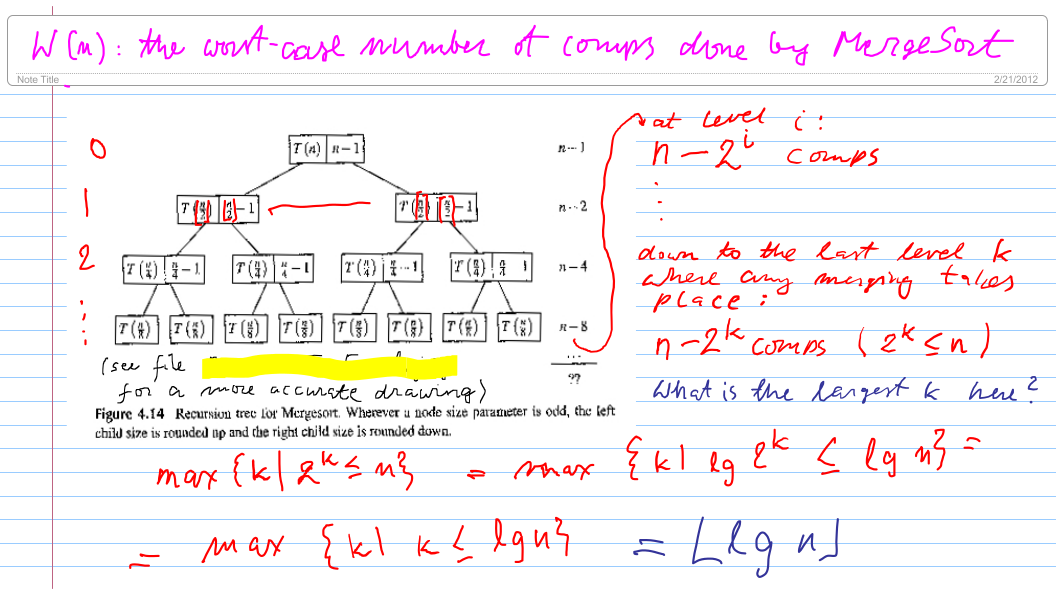} 

\includegraphics[width=0.9\linewidth]{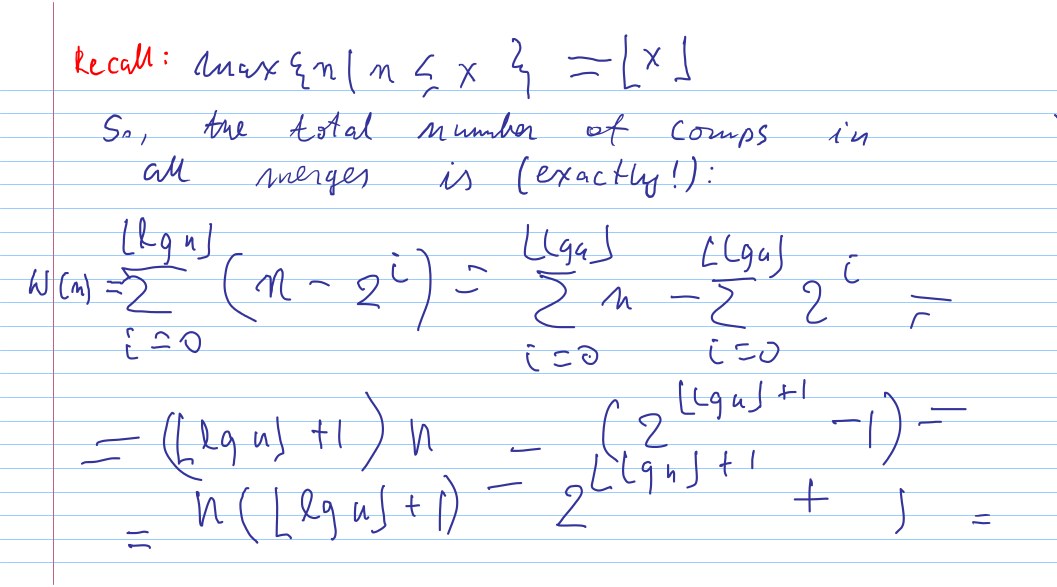} 

\includegraphics[width=0.9\linewidth]{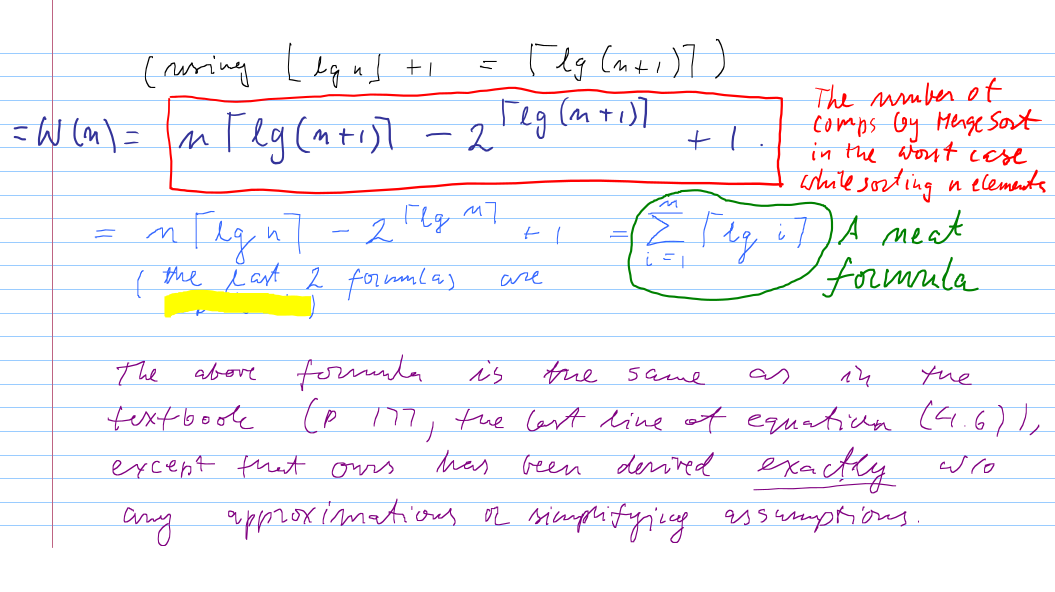}

\includegraphics[width=0.9\linewidth]{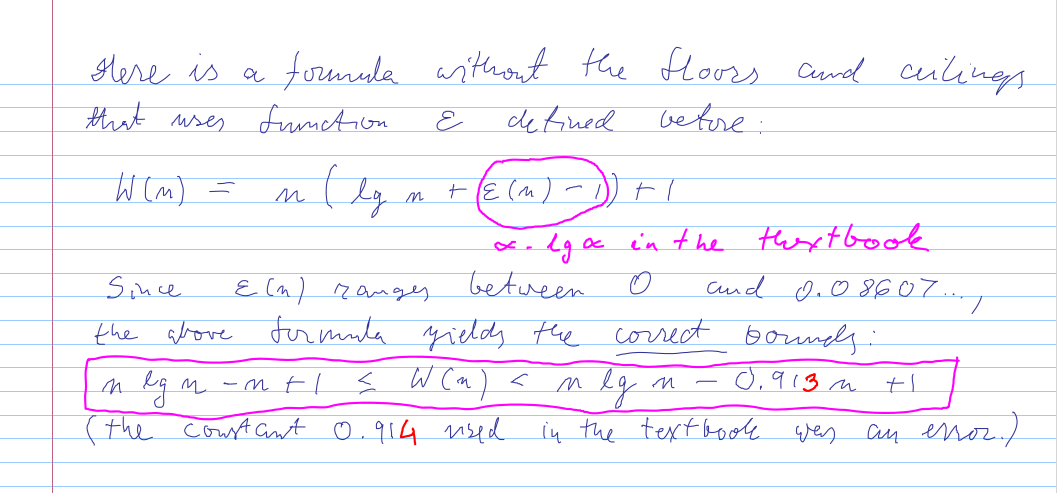}

\bigskip

Below is the improved recursion tree (of the Figure~4.14 page page 177 of \cite{baavan:ana}) that I used in class in Spring 2012.
 \\
\includegraphics[width=1.1\linewidth]{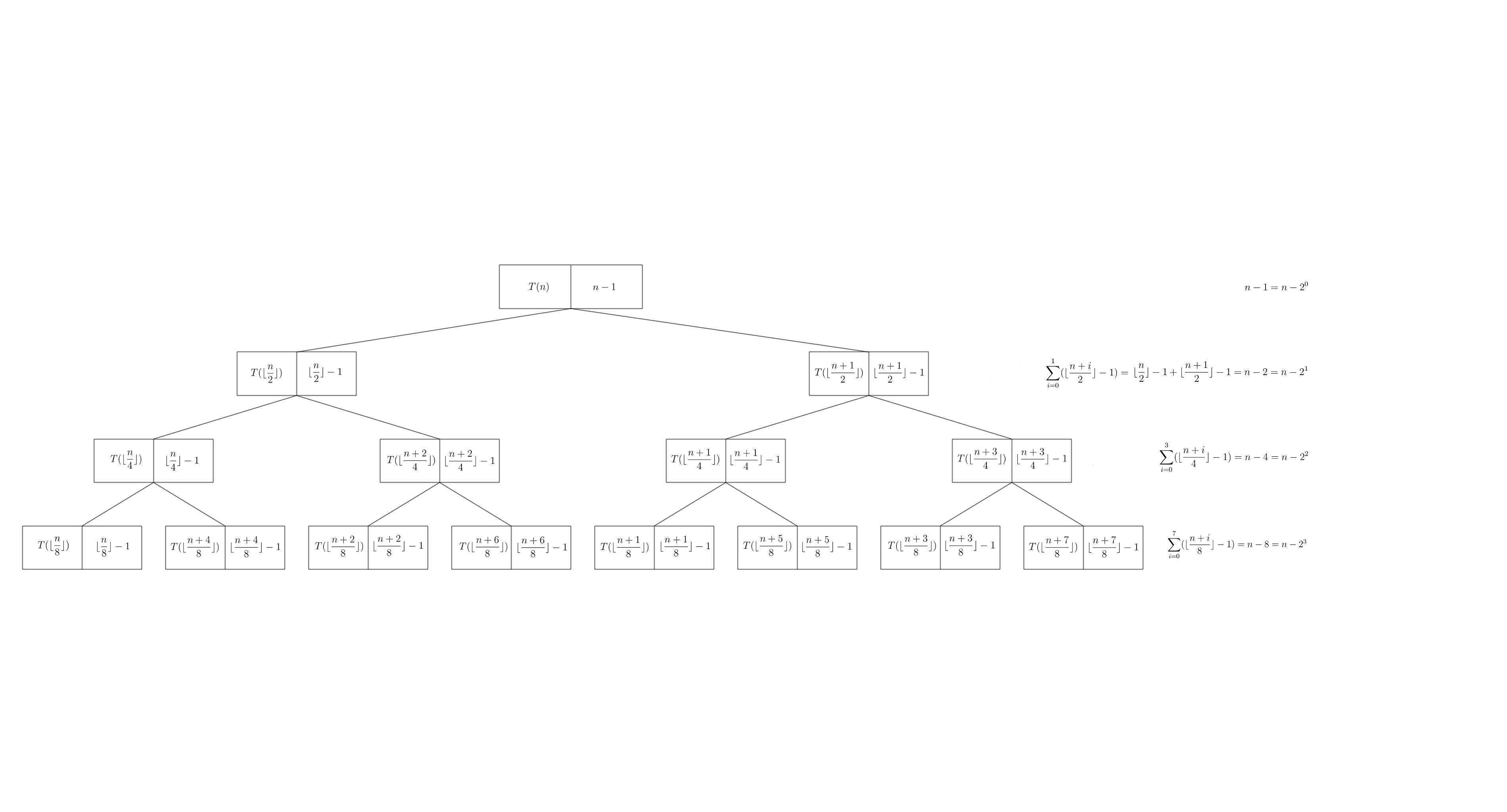}

In Spring 2010 and before, I was deriving the equality \eqref{eq:main_level_comps_formula} on page~\pageref{eq:main_level_comps_formula} during my lectures directly from the recurrence relation \eqref{eq:recmergesort1}, \eqref{eq:recmergesort2} on page~\pageref{eq:recmergesort1}.

\bibliographystyle{siam}
\bibliography{ref}
\bigskip
\bigskip
\copyright \textit{2017 Marek A. Suchenek. All rights reserved by the author. \newline A non-exclusive license to distribute this article is granted to arXiv.org}.

\end{document}